\documentclass[aps,prl,reprint,showpacs,groupedaddress]{revtex4-1}
\usepackage{CJK}

\usepackage{amsmath,amssymb,amsthm}        
\usepackage{mathrsfs}

\newcommand{\re}{{\rm e}}
\newcommand{\ri}{{\rm i}}
\newcommand{\rd}{{\rm d}}
\newcommand{\rD}{{\rm D}}

\newcommand{\norm}[1]{\lVert#1\rVert}
\newtheorem*{theorem}{Theorem}
\newtheorem*{definition}{Definition}

\begin{document}
\begin{CJK*}{GBK}{}


\title{Geometric phases between biorthogonal states}

\author{Xiao-Dong Cui}
\affiliation{ School of Physics, Shandong University, Jinan 250100, China}

\author{Yujun Zheng }
\email{electronic mail: yzheng@sdu.edu.cn}
\affiliation{ School of Physics, Shandong University, Jinan 250100, China}

\pacs{03.65.Vf}
%
\begin{abstract}
We investigate the evolution of a state which is dominated by a
finite-dimensional non-Hermitian time-dependent Hamiltonian operator
with a nondegenerate spectrum by using a biorthonormal approach.
The geometric phase between any two states, biorthogonal or not,
are generally derived by employing the generalized interference method.
The counterpart of Manini-Pistolesi non-diagonal geometric phase in
the non-Hermitian setting is taken as a typical example.
\end{abstract}

\maketitle

One of the axioms in conventional Hermitian quantum mechanics requires that
the Hamiltonian and other observables are represented by Hermitian
operators for mathematical convenience and physical reasonableness \cite{MFQM}.
This axiom has, however, been challenged by the $\mathcal{PT}$ and $\mathcal{CPT}$ symmetric properties
of Hamiltonians in quantum field theory \cite{PhysRevLett.80.5243,PhysRevLett.89.270401}.
The formulation of non-Hermitian quantum mechanics tolerates physicists to enlarge
the set of possible and even non-Hermitian Hamiltonian to describe novel physical
phenomena effectively and provide powerful numerical and analytical methods \cite{NHQM}.

As an important measurable physical quantity in conventional Hermitian quantum system,
geometric phase \cite{Proc.Ind.Acad.Sci.A.44.247,Proc.R.Soc.Lond.A.392.45} not
only has many theoretical derivatives \cite{PhysRevLett.52.2111,PhysRevLett.58.1593,
PhysRevLett.60.2339,AnnsPhys.228.205,PhysRevLett.85.3067,PhysRevLett.94.070406}
but also attracts much attention due to its potential application in quantum computation
\cite{PhysLettA.264.94,PhysRevLett.95.130501,PhysRevLett.109.170501,PhysRevLett.110.190501}.
The issue of geometric phase in non-Hermitian quantum systems was
considered firstly by Garrison and Wright \cite{Phys.Lett.A.128.177} and
subsequently by other contributors \cite{J.Phys.A:Math.Gen.23.5795,Europhys.Lett.13.199,
J.Phys.A:Math.Gen.29.2567,Phys.Lett.A.264.11,J.Math.Phys.49.082105,PhysRevA.86.064104} by
using a biorthonormal eigenbasis procedure.
Geometric phase is thence generalized in the complex-valued field, the real part
of which is the geometric phase of usual sense while the imaginary part makes
possible to investigate the geometric dilation or contraction of the modulo of
a wavefunction.
However, if the initial and final states are biorthogonal,
one can not extract any complex-valued phase information directly from their
inner product.
The geometric phases between biorthogonal states have thus not been defined yet.
In this manuscript standing on the geometric point of view,
we firstly re-define the concept of ``in-phase'' based on the generalized interference method, and hereafter generalize the relevant geometric quantities which will provide a deep understanding on geometric phases in a class of non-Hermitian quantum systems, and finally suggest a general formula for geometric phase between any two states, biorthogonal or not.

We consider an $N$-dimensional time-dependent non-Hermitian Hamiltonian $H(t)$ with
a non-degenerate instantaneous spectrum. There exists a complete biorthonormal set of basis vectors $\{|n(t)\rangle,~|\tilde{n}(t)\rangle:~n=1,2,\cdots,N\}$ obeying $H(t)|n(t)\rangle=E_n(t)|n(t)\rangle,~H^\dag(t)|\tilde{n}(t)\rangle=E_n^*(t)|\tilde{n}(t)\rangle$ \cite{J.Math.Phys.8.2039}, such that
\begin{eqnarray}
\label{bion}
\langle\tilde{m}(t)|n(t)\rangle=\delta_{mn},\quad \sum_{n=1}^N|n(t)\rangle\langle\tilde{n}(t)|=1.
\end{eqnarray}
In order to facilitate the subsequential deduction, we here define
the operation of tilde satisfying
\begin{eqnarray}
\label{tilde}
  \tilde~:~|\psi\rangle\rightarrow|\tilde{\psi}\rangle,\qquad(~\tilde{}~)^2=id,
\end{eqnarray}
The time parameter $t$ in the complete biorthonormal basis implies that this is a moving frame,
such that for any nonzero vector $|\psi(t)\rangle\in\overline{{\rm Span}(\{|n(t)\rangle\}_{n=1}^N)}$,
there exist $1$ binormalized vector $|\tilde{\psi}^\|(t)\rangle\in\overline{{\rm Span}(\{|\tilde{n}(t)\rangle\}_{n=1}^N)}$
and $N-1$ biorthogonal rays $\tilde{\psi}^\bot(t)\in\overline{{\rm Span}(\{|\tilde{n}(t)\rangle\}_{n=1}^N)}/\mathbb{C}-\{0\}$,
{\it i.e.}, $\langle\tilde{\psi}^\|(t)|\psi(t)\rangle=1,~\langle\tilde{\psi}^\bot(t)|\psi(t)\rangle=0,~
\langle\tilde{\psi}^\bot(t)|\psi^\bot(t)\rangle=1$.

Also, the state evolution of a non-Hermitian quantum system obeys the
non-Hermitian Schr\"{o}dinger equation ($\hbar=1$)
\begin{equation}
\label{sch}
H(t)|\psi(t)\rangle = \ri\frac{\rd}{\rd t}|\psi(t)\rangle,
\end{equation}
where $H(t)\neq H^{\dag}(t)$ is a $N$-dimensional matrix representation.
The non-hermiticity of $H(t)$ means that the evolution operator $U(t)$ is
non-unitary and it results in $\langle\psi(t)|\psi(t)\rangle\neq\langle\psi(0)|\psi(0)\rangle,~\forall t>0$.
Usually, it is helpful to introduce the adjoint Schr\"{o}dinger equation
corresponding to Eq.~(\ref{sch}),
\begin{equation}
\label{adsch}
H^{\dag}(t)|\tilde{\psi}(t)\rangle = \ri\frac{\rd}{\rd t}|\tilde{\psi}(t)\rangle.
\end{equation}
Based on Eqs.~(\ref{sch}) and (\ref{adsch}), one can obtain the following relations
\begin{equation}
\frac{\rd}{\rd t}\Big(\langle\tilde{\psi}(t)|\psi(t)\rangle\Big)=0.
\end{equation}
and the binormalization of quantum states for non-Hermitian quantum system
\begin{equation}
\label{bin}
\langle\tilde{\psi}(t)|\psi(t)\rangle=\langle\tilde{\psi}(0)|\psi(0)\rangle=1,~\forall t\geqslant0.
\end{equation}
The binormalization is invariant under local gauge transformation or complex
scaling transformation, namely,
\begin{eqnarray}
\label{gt}
\begin{array}{lll}
|\psi\rangle&\mapsto&|\psi'\rangle=\re^{\ri\zeta}|\psi\rangle, \\
|\tilde{\psi}\rangle&\mapsto&|\tilde{\psi}'\rangle=\re^{\ri\tilde{\zeta}}|\tilde{\psi}\rangle,
\end{array}
\end{eqnarray}
where $\zeta,\tilde{\zeta}\in\mathbb{C}$ are complex numbers, and they satisfy
the requirement $\zeta=\tilde{\zeta}^{*}$.
It should be stressed here that only the state vector $|\psi(t)\rangle\in\overline{{\rm Span}(\{|n(t)\rangle\}_{n=1}^N)}$ rather than its dual $|\tilde{\psi}(t)\rangle\in\overline{{\rm Span}(\{|\tilde{n}(t)\rangle\}_{n=1}^N)}$ describes a quantum state of the physical system, although they stand equally from Eqs.~(\ref{sch}) and (\ref{adsch}).

In conventional Hermitian quantum mechanics, the maximum interference formula
between any two non-orthogonal normalized rays $A,B$ is written by
\begin{eqnarray}
\label{if}
I_{max}&=&\sup_{\alpha\in\mathbb{R}}\norm{A\re^{\ri\alpha}+B}^2\nonumber\\
&=&\sup_{\alpha\in\mathbb{R}} \langle A\re^{-\ri\alpha}+B|A\re^{\ri\alpha}+B\rangle,
\end{eqnarray}
which induces the Pancharatnam connection
$\mathcal{A}^{P}={\rm Im}\langle A|B\rangle$ \cite{PhysRevLett.60.2339}.
In order to investigate the issue of geometric phase in the above
non-Hermitian setting, the Pancharatnam connection need modifying by generalizing
the maximum interference formula Eq.~(\ref{if}).
According to the bi-normalization Eq.~(\ref{bin}) with local gauge
transformation Eq.~(\ref{gt}), the generalized interference formula
$\mathfrak{I}^{2}$ between any two non-biorthonormal rays $\psi_1,\psi_2$ is defined as \cite{PhysRevA.86.064104}
\begin{eqnarray}
\label{gif}
\mathfrak{I}^2 &=& \langle\tilde{\psi}_1\re^{-\ri\tilde{\theta}^{*}}+
    \tilde{\psi}_2|\psi_1\re^{\ri\theta}+\psi_2\rangle \nonumber\\
  &=& \langle\tilde{\psi}_1\re^{-\ri\theta}+\tilde{\psi}_2|\psi_1\re^{\ri\theta}+\psi_2\rangle,
    \quad\theta\in\mathbb{C}.
\end{eqnarray}
\begin{definition}
When the generalized interference intensity $\mathfrak{I}^{2}$ is
stationary with respect to the complex-valued phase $\theta$ such that $\sqrt{\frac{\langle\tilde{\psi}_1|\psi_2\rangle}{\langle\tilde{\psi}_2|\psi_1\rangle}}=1$,
then $|\psi_1\rangle$ and $|\psi_2\rangle$ are said to be ``in phase'' or parallel.
\end{definition}
According to the definition, the generalized Pancharatnam connection $\mathcal{A}^{GP}$ is given by
\begin{equation}
\label{gpc}
\mathcal{A}^{GP}=\sqrt{\frac{\langle\tilde{\psi}_1|\psi_2\rangle}{\langle\tilde{\psi}_2|\psi_1\rangle}}-1.
\end{equation}
As the infinitesimal version of the real-valued Pancharatnam
connection ${\rm Im}\langle B(s)|B(s+\rd s)\rangle=0$ can induce a parallel
transport law ${\rm Im}\langle B(s)|\frac{\rd}{\rd s}B(s)\rangle=0$ by
\begin{equation}
\langle B(s)|B(s+\rd s)\rangle=1+\Big\langle B(s) \Big| \frac{\rd}{\rd s}B(s) \Big\rangle\rd s+
    \mathcal{O}(\rd s^2),
\end{equation}
the counterpart of the generalized Pancharatnam connection
$\sqrt{\frac{\langle\tilde{\psi}(s)|\psi(s+\rd s)\rangle}{\langle\tilde{\psi}(s+\rd s)|\psi(s)\rangle}}-1=0$
can also give a parallel transport law in the non-Hermitian setting
\begin{equation}
\Big\langle\tilde{\psi}(s)\Big|\frac{\rd}{\rd s}\psi(s)\Big\rangle=0,\quad\Big\langle\frac{\rd}{\rd s}\tilde{\psi}(s)\Big|\psi(s)\Big\rangle=0 ,
\end{equation}
by
\begin{eqnarray}
\sqrt{\frac{\langle\tilde{\psi}(s)|\psi(s+\rd s)\rangle}{\langle\tilde{\psi}(s+\rd s)|\psi(s)\rangle}}
&=& 1+ \Big\langle\tilde{\psi}(s) \Big|\frac{\rd}{\rd s}\psi(s) \Big\rangle\rd s+
    \mathcal{O}(\rd s^2)\nonumber\\
&=& 1- \Big\langle\frac{\rd}{\rd s}\tilde{\psi}(s) \Big|\psi(s) \Big\rangle\rd s+
    \mathcal{O}(\rd s^2),\nonumber\\
\end{eqnarray}
here bi-normalization Eq.~(\ref{bin}) has been used.
It should be noted that the parallel transport law
$\sqrt{\frac{\langle\tilde{\psi}(s)|\psi(s+\rd s)\rangle}{\langle\tilde{\psi}(s+\rd s)|\psi(s)\rangle}}-1$
must be equal to $0$ rather than any other complex numbers,
because any complex number (except $0$) can be expressed as an exponential
of a complex number which will be involved into the complex-valued phase $\theta$.
Moreover, the infinitesimal version of the generalized Pancharatnam
connection Eq.~(\ref{gpc}) gives
\begin{eqnarray}
\mathcal{A}^{GP}(s)&=&\Big\langle\tilde{\psi}(s)\Big|\frac{\rd}{\rd s}\psi(s)\Big\rangle,
\end{eqnarray}
which transform under the gauge transformation Eq.~(\ref{gt}) as follow,
\begin{eqnarray}
\mathcal{A}^{GP}(s) &\mapsto& \mathcal{A}^{GP}(s)+\ri\frac{\rd\zeta}{\rd s}.
\end{eqnarray}
The tangent vector $|\frac{\rd}{\rd s}\psi(s)\rangle$ is not gauge covariant,
\begin{eqnarray}
 \Big|\frac{\rd}{\rd s}\psi(s) \Big\rangle &\mapsto& \re^{\ri\zeta} \bigg( \Big|\frac{\rd}{\rd s}\psi(s) \Big\rangle+\ri\frac{\rd\zeta}{\rd s} \Big|\psi(s) \Big\rangle\bigg).
\end{eqnarray}

One can check that the covariant derivative $\frac{\rD}{\rd s}$ can be defined as
\begin{eqnarray}
\label{cd}
\frac{\rD}{\rd s}|\psi(s)\rangle &=& \Big(\frac{\rd}{\rd s}-
           \mathcal{A}^{GP}(s)\Big)|\psi(s)\rangle.
\end{eqnarray}
Likewise, the duals of Eqs.~(\ref{gif})-(\ref{cd}) can be obtained by the operation of tilde Eq.~(\ref{tilde}).
Hence, there exists a gauge invariant quantity $\langle\frac{\tilde{\rD}}{\rd s}\tilde{\psi}(s)|\frac{\rD}{\rd s}\psi(s)\rangle$
which can be used to defined a metric on ray space,
\begin{equation}
\label{metric}
\rd\mathfrak{L}^2=\Big\langle\frac{\tilde{\rD}}{\rd s}\tilde{\psi}(s)\Big| \frac{\rD}{\rd s}\psi(s)\Big\rangle\rd s^2.
\end{equation}
The metric Eq.~(\ref{metric}) then determines the geodesic in ray space by
variation of the length $I(\mathcal{C})$
\begin{eqnarray}
I(\mathcal{C})=\int_\mathcal{C}\rd\mathfrak{L}=\int_\mathcal{C}\sqrt{\rd \mathfrak{L}^2}=\int_\mathcal{C}\sqrt{\Big\langle\frac{\tilde{\rD}}{\rd s}\tilde{\psi}(s)\Big|\frac{\rD}{\rd s}\psi(s)\Big\rangle}~\rd s,\nonumber\\
\end{eqnarray}
from which one can obtain a pair of geodesic equations,
\begin{eqnarray}
\label{gd}
\frac{\rD^2}{\rd s^2}|\psi(s)\rangle=0 ,\quad
\frac{\tilde{\rD}^2}{\rd s^2}|\tilde{\psi}(s)\rangle=0.
\end{eqnarray}
Equations in Eq.~(\ref{gd}) are gauge covariant under local gauge transformation Eq.~(\ref{gt}).
It should be stressed that the equations in Eq.~(\ref{gd}) must hold simultaneously.
When the generalized interference intensity $\mathfrak{I}^2$ in Eq.~(\ref{gif})
is stationary with respect to the complex-valued phase $\theta$,
the generalized Pancharatnam phase can be obtained,
\begin{equation}
\theta_{1,2}^{GP}=-\frac{\ri}{2}\log\frac{\langle\tilde{\psi}_1|\psi_2\rangle}{\langle\tilde{\psi}_2|\psi_1\rangle}.
\end{equation}

\begin{theorem}
Let the two non-biorthogonal states $|\psi_1\rangle,|\psi_2\rangle$ be connected by a geodesic $G^{1,2}$ satisfying Eq.~(\ref{gd}), then the generalized Pancharatnam phase $\theta_{1,2}^{GP}$ is given by
\begin{equation}
\label{proof}
\theta_{1,2}^{GP}=-\ri\int_{G^{1,2}}\mathcal{A}^{GP}£¬\nonumber
\end{equation}
where $\mathcal{A}^{GP}=\langle\tilde{\psi}|\rd\psi\rangle$ is the connection 1-form.
\end{theorem}
\begin{proof}Consider a geodesic $|\varphi(s)\rangle$ starting from $|\varphi(0)\rangle=|\psi_1\rangle$ and ending in the ray $\psi_2\ni|\varphi(1)\rangle$ satisfying $\mathcal{A}^{GP}(s)=0$, then geodesic equation Eq.~(\ref{gd}) reduces to $\frac{\rd^2}{\rd s^2}|\varphi(s)\rangle=0$, whose solution is a straight line described by
\begin{eqnarray}
\label{sl}
|\varphi(s)\rangle=(1-s)|\psi_1\rangle+s|\varphi(1)\rangle,~\forall s\in[0,1].
\end{eqnarray}
Let $q(s)=\sqrt{\frac{\langle\tilde{\psi}_1|\varphi(s)\rangle}{\langle\tilde{\varphi}(s)|\psi_1\rangle}}-1$.
It can be verified that $q(0)=0$ and $\dot{q}(0)=0$ due to $\mathcal{A}^{GP}(s)=0$ on the geodesic $|\varphi(s)\rangle$.
By inserting Eq.~(\ref{sl}) into $\dot{q}(0)=0$, one can obtain
\begin{eqnarray}
\label{csl}
\langle\tilde{\psi}_1|\varphi(1)\rangle=\langle\tilde{\varphi}(1)|\psi_1\rangle,
\end{eqnarray}
where we use the binormalization condition
 $\langle\tilde{\psi}_1|\psi_1\rangle=\langle\psi_1|\tilde{\psi}_1\rangle=1$.
By inserting Eqs.~(\ref{sl}) and (\ref{csl}) into $q(s)$,
then we have $q(s)=0,~\forall s\in[0,1]$, which means $|\psi_1\rangle$
and $|\varphi(s)\rangle$ are ``in phase''.
Due to the gauge covariance of the geodesic equation Eq.~({\ref{gd}}),
let $|\psi(s)\rangle=\re^{\ri\theta(s)}|\varphi(s)\rangle$,
$\forall s\in[0,1]$, with the boundary condition $\theta(0)=0$
and $\theta(1)=\theta_{1,2}^{GP}$, then $|\psi(s)\rangle$ is
still a geodesic linking $|\psi_1\rangle$ to $|\psi_2\rangle$,
which is denoted by $G^{1,2}$.
And finally, $-\ri\int_{G^{1,2}}\mathcal{A}^{GP}=
\int_0^1\frac{\rd}{\rd s}\theta(s)\rd s=\theta_{1,2}^{GP}$. 
\end{proof}
According to the theorem, one can link
$N$ vertices $|\psi_1\rangle,|\psi_2\rangle,\cdots,|\psi_N\rangle$
one-by-one by $N-1$ geodesics $G^{1,2},G^{2,3},\cdots,G^{N-1,N}$ to
obtain the accumulated generalized Pancharatnam phase
$\theta_{1,2,\cdots,N}^{GP}$ along the continuous curve
$G^{open}=G^{1,2}+G^{2,3}+\cdots+G^{N-1,N}$ by
\begin{eqnarray}
\theta_{1,2,\cdots,N}^{GP} &=& -\ri\int_{G^{open}}\mathcal{A}^{GP}\nonumber\\
  &=&-\ri\sum_{n=1}^{N-1}\int_{G^{n,n+1}}\mathcal{A}^{GP}\nonumber\\
  &=&\sum_{n=1}^{N-1}\theta_{n,n+1}^{GP}.
\end{eqnarray}

It should be noted that $\theta_{1,2,\cdots,N}^{GP}$ is not
gauge invariant under local gauge transformation Eq.~(\ref{gt})
and thence involves in any possible phase including dynamical phase.
However, if $|\psi_N\rangle$ is linked back to $|\psi_1\rangle$ by a geodesic $G^{N,1}$, then the curve $G^{closed}=G^{1,2}+G^{2,3}+\cdots+G^{N-1,N}+G^{N,1}$
is continuous and closed such that $\theta_{1,2,\cdots,N}^{GP}$ is gauge invariant,
\begin{eqnarray}
\theta_{1,2,\cdots,N,1}^{GP} &=& -\ri\oint_{G^{closed}}\mathcal{A}^{GP}\nonumber\\
  &=& -\ri\sum_{n=1}^{N}\int_{G^{n,n+1\bmod N}}\mathcal{A}^{GP}\nonumber\\
  &=& \sum_{n=1}^{N}\theta_{n,n+1\bmod N}^{GP}.
\end{eqnarray}
Here, the geodesic $G^{N,1}$ is added to remove any possible phase,
including dynamical phase, which can be produced or removed by local
gauge transformation Eq.~(\ref{gt}).
Hence, the accumulated generalized Pancharatnam phase $\theta_{1,2,\cdots,N}^{GP}$
is purely geometrical.
As $N\rightarrow\infty$, the curve $G^{open}$ becomes smooth while $G^{N,1}$ is still unchanged.
The geometric phase difference between
$|\psi(1)\rangle=|\psi_N\rangle$ and $|\psi(0)\rangle=|\psi_1\rangle$,
 $\lim_{N\rightarrow\infty}\theta_{1,2,\cdots,N,1}^{GP}$,
can be calculated by
\begin{eqnarray}
\label{nvgp}
&&\lim_{N\rightarrow\infty}\theta_{1,2,\cdots,N,1}^{GP}\nonumber\\
&=& \lim_{N\rightarrow\infty}\sum_{n=1}^{N}\theta_{n,n+1\bmod N}^{GP}\nonumber\\
&=&\lim_{N\rightarrow\infty}\left\{-\frac{\ri}{2}\log\left[\frac{\langle\tilde{\psi}_{1}|
  \psi_{2}\rangle}{\langle\tilde{\psi}_{2}|\psi_{1}\rangle}
\cdots\frac{\langle\tilde{\psi}_{N-1}|\psi_{N}\rangle}
   {\langle\tilde{\psi}_{N}|\psi_{N-1}\rangle}\frac{\langle\tilde{\psi}_{N}|
   \psi_{1}\rangle}{\langle\tilde{\psi}_{1}|\psi_{N}\rangle}\right]\right\}\nonumber\\
&=&\lim_{N\rightarrow\infty}\left\{-\frac{\ri}{2}\log\left[\frac{\langle\tilde{\psi}_{N}|
   \psi_{1}\rangle}{\langle\tilde{\psi}_{1}|\psi_{N}\rangle}\cdot\frac{\prod_{n=1}^{N}\langle\tilde{\psi}_{n}|\psi_{n}+
   \frac{\rd}{\rd s}\psi_{n}\Delta s\rangle}{\prod_{n=1}^{N}\langle\tilde{\psi}_{n}+
   \frac{\rd}{\rd s}\tilde{\psi}_{n}\Delta s|\psi_{n}\rangle}\right]\right\}\nonumber\\
&=&-\frac{\ri}{2}\log\frac{\langle\tilde{\psi}(1)|\psi(0)\rangle}{\langle\tilde{\psi}(0)|\psi(1)\rangle}-
  \ri\int_0^1 \Big\langle\tilde{\psi}(s) \Big|\frac{\rd}{\rd s}\psi(s) \Big\rangle\rd s\nonumber\\
&=&-\lim_{N\rightarrow\infty}\theta_{N,N-1,\cdots,1,N}^{GP}  ,
\end{eqnarray}
where $\lim_{N\rightarrow\infty}\theta_{N,N-1,\cdots,1,N}^{GP}$ represents
the geometric phase difference between $|\psi(0)\rangle$ and $|\psi(1)\rangle$.

Based on the above preparation, we now discuss the geometric
phase between initial and final states no matter whether they are
non-biorthogonal or biorthogonal, $i.e.$, $\langle\tilde{\psi}(0)|\psi(t)\rangle\neq0$ or
$\langle\tilde{\psi}(0)|\psi(t)\rangle=0$.

{\it Non-biorthogonal case} --- The evolving state $|\psi(t)\rangle$ starting
from the initial state $|\psi(0)\rangle$ is governed by the non-Hermitian
Schr\"{o}dinger equation Eq.~(\ref{sch}).
The geometric phase $\gamma^{geo}(0,t)$ between $|\psi(0)\rangle$
and $|\psi(t)\rangle$ can be calculated by Eq.~(\ref{nvgp}),
\begin{eqnarray}
\label{nogp}
\gamma^{geo}(0,t)&=& -\frac{\ri}{2}\log\frac{\langle\tilde{\psi}(0)|\psi(t)\rangle}{\langle
   \tilde{\psi}(t)|\psi(0)\rangle}+\ri\int_0^t \Big\langle\tilde{\psi}(t')
      \Big|\frac{\rd}{\rd t'}\psi(t') \Big\rangle\rd t' \nonumber\\
  &=& -\frac{\ri}{2}\log\frac{\langle\tilde{\psi}(0)|\psi(t)\rangle}{\langle\tilde{\psi}(t)|\psi(0)\rangle}+\int_0^t\langle\tilde{\psi}(t')|H(t)|\psi(t')\rangle\rd t' \nonumber\\
  &=& \theta^{GP}(0,t)-\gamma^{dyn}(0,t).
\end{eqnarray}
It should be noted in Eq.~(\ref{nogp}) that the existence of $\gamma^{geo}(0,t)$
as well as $\theta^{GP}(0,t)$ merely depends on whether the initial state
is biorthogonal to the final state rather than any intermediately traveled state,
while the dynamical phase $\gamma^{dyn}(0,t)$ continuously exists.

{\it Biorthogonal case} --- Due to the biorthogonality between the initial
state $|\psi(0)\rangle$ and the final state $|\psi(t)\rangle$,
{\it i.e.}, $\langle\tilde{\psi}(0)|\psi(t)\rangle=0$,
the geometric phase between them can not be evaluated by Eq.~(\ref{nogp})
directly for $\log0$ is not defined mathematically.
However, if an intermediately traveled state $|\psi(t_1)\rangle$ is
non-biorthogonal to both the initial and the final states,
then the geometric phase $\gamma^{geo}(0,t)$ between the
initial state $|\psi(0)\rangle$ and the final state $|\psi(t)\rangle$
can still be calculated indirectly by Eq.~(\ref{nogp}),
\begin{eqnarray}
\label{ogp}
\gamma^{geo}(0,t) &=& \gamma^{geo}(0,t_1)+\gamma^{geo}(t_1,t) \nonumber\\
   &=& -\frac{\ri}{2}\log\frac{\langle\tilde{\psi}(0)|\psi(t_1)\rangle\langle\tilde{\psi}(t_1)
   |\psi(t)\rangle}{\langle\tilde{\psi}(t)|\psi(t_1)\rangle\langle\tilde{\psi}(t_1)|\psi(0)\rangle} \nonumber\\
   & &\qquad+\int_{0}^{t}\langle\tilde{\psi}(t')|H(t)|\psi(t')\rangle\rd t'.
\end{eqnarray}
Here, the intermediately traveled state $|\psi(t_1)\rangle$ acts as
a torchbearer to guarantee that the geometric phase difference can
be preserved and delivered from the initial state to the final biorthogonal state.
Besides, $|\psi(t_1)\rangle$ does not interrupt the process of the state evolution.
Seen from another perspective, both the initial and the final states are
projected onto the intermediately traveled state, and the total geometric
phase difference is equal to the difference between $\gamma^{geo}(0,t_1)$
and $\gamma^{geo}(t,t_1)$,
\begin{equation}
\label{ogp2}
\gamma^{geo}(0,t)=\gamma^{geo}(0,t_1)-\gamma^{geo}(t,t_1).
\end{equation}
Hence, the intermediately traveled state $|\psi(t_1)\rangle$ is
unnecessary because it can be replaced with any state $|a\rangle$
which is non-biorthogonal to both the initial and final states to
implement Eq.~(\ref{ogp2}),
\begin{eqnarray}
\label{ogp3}
\gamma^{geo}(0,t) &=& -\frac{\ri}{2}\log\frac{\langle\tilde{\psi}(0)|a\rangle\langle\tilde{a}
   |\psi(t)\rangle}{\langle\tilde{\psi}(t)|a\rangle\langle\tilde{a}|\psi(0)\rangle} \nonumber\\
   & &\qquad+\int_{0}^{t}\langle\tilde{\psi}(t')|H(t)|\psi(t')\rangle\rd t'.
\end{eqnarray}
The first term in Eq.~(\ref{ogp3}) can be obtained by modifying Eq.~(\ref{gif}),
\begin{equation}
\mathfrak{I}^2=\Big\langle\re^{-\ri\theta}\langle\tilde{\psi}(0)|a\rangle\tilde{a}+
  \langle\tilde{\psi}(t)|a\rangle\tilde{a}~\Big|~a\langle\tilde{a}|\psi(0)\rangle\re^{\ri\theta}+
  a\langle\tilde{a}|\psi(t)\rangle\Big\rangle.
\end{equation}
The second term in Eq.~(\ref{ogp3}) is to remove the dynamical
phase off the final state $|\psi(t)\rangle$.

As a typical example, we consider the off-diagonal geometric phases \cite{PhysRevLett.85.3067}
in the non-Hermitian setting by using Eq.~(\ref{ogp3}):
two states $|j(0)\rangle$ and $|k(0)\rangle$ evolve adiabatically to $|j(t)\rangle$ and $|k(t)\rangle$,
respectively, such that $\langle\tilde{j}(0)|j(t)\rangle=0$ and $\langle\tilde{k}(0)|k(t)\rangle=0$.
We can find a state $|a\rangle$ which is not biorthogonal to $|j(0)\rangle$,
$|j(t)\rangle$, $|k(0)\rangle$, or $|k(t)\rangle$.
Then the off-diagonal geometric phases $\gamma^{geo}_{jk}$ is given by
\begin{eqnarray}
\gamma^{geo}_{jk}&=&\gamma^{geo}[|j(0)\rangle,|a\rangle,|k(t)\rangle]+\gamma^{geo}[|k(0)\rangle,|a\rangle,|j(t)\rangle]\nonumber\\
&&+\gamma^{geo}_{j}(0,t)+\gamma^{geo}_k(0,t),
\end{eqnarray}
where
\begin{eqnarray}
\gamma^{geo}[|j(0)\rangle,|a\rangle,|k(t)\rangle]=-\frac{\ri}{2}
 \log\frac{\langle\tilde{j}(0)|a\rangle\langle\tilde{a}|k(t)\rangle\langle\tilde{k}(t)|j(0)
 \rangle}{\langle\tilde{a}|j(0)\rangle\langle\tilde{k}(t)|a\rangle\langle\tilde{j}(0)|k(t)\rangle} , \nonumber\\
\gamma^{geo}[|k(0)\rangle,|a\rangle,|j(t)\rangle]=-\frac{\ri}{2}
\log\frac{\langle\tilde{k}(0)|a\rangle\langle\tilde{a}|j(t)\rangle\langle\tilde{j}(t)|k(0)
\rangle}{\langle\tilde{a}|k(0)\rangle\langle\tilde{j}(t)|a\rangle\langle\tilde{k}(0)|j(t)\rangle} ,\nonumber\\
\end{eqnarray}
and
\begin{eqnarray}
\gamma^{geo}_{j}(0,t)&=&-\frac{\ri}{2}\log\frac{\langle\tilde{j}(0)|a\rangle\langle\tilde{a}
   |j(t)\rangle}{\langle\tilde{j}(t)|a\rangle\langle\tilde{a}|j(0)\rangle} \nonumber\\
   & &\qquad+\int_{0}^{t}\langle\tilde{j}(t')|H(t)|j(t')\rangle\rd t' ,\nonumber\\
\gamma^{geo}_{k}(0,t)&=&-\frac{\ri}{2}\log\frac{\langle\tilde{k}(0)|a\rangle\langle\tilde{a}
   |k(t)\rangle}{\langle\tilde{k}(t)|a\rangle\langle\tilde{a}|k(0)\rangle} \nonumber\\
   & &\qquad+\int_{0}^{t}\langle\tilde{k}(t')|H(t)|k(t')\rangle\rd t'.
\end{eqnarray}
For more than two biorthogonal states, the similar procedure can be performed.

In conclusion, we investigated and also suggested a general formalism for the geometric phases between any two states, biorthogonal or not, in a finite-dimensional non-Hermitian quantum dynamical system with a non-degenerate spectrum.
Based on the generalized interference formula, we also re-defined the concept of ``in-phase'' in the non-Hermitian setting, which contributed to the discussion of geometric aspects.
Finally, we gave the counterpart of Manini-Pistolesi non-diagonal geometric phase in the non-Hermitian setting as a typical example.

This work was supported by the National Science Foundation
of China (Grants No. 91021009 and No. 21073110).


\end{CJK*}


\end{document}